\documentclass[preprint,showpacs,preprintnumbers,amsmath,amssymb]{revtex4}

\usepackage{amsmath}
\usepackage{latexsym}
\usepackage{amsfonts}
\usepackage{amssymb}
\usepackage{amsthm}

\usepackage{graphicx}
\usepackage{dcolumn}
\usepackage{bm}

\theoremstyle{plain}

\newtheorem{theorem}{Theorem}
\newtheorem{lemma}{Lemma}
\newtheorem{corollary}{Corollary}
\theoremstyle{definition}

\newtheorem{remark}{Remark}
\newtheorem{example}{Example}

\renewenvironment{proof}[1][Proof]{\textbf{#1.} \ }{\hspace{\stretch{1}}\rule{0.5em}{0.5em} \\} 

\newcommand{\eps}{\varepsilon}

\newcommand{\hh}{\mathcal{H}}

\newcommand{\lh}{\mathcal{B}(\mathcal{H})} 
\newcommand{\bbb}[1]{\mathcal{B}_{#1}(\mathcal{H})} 
\newcommand{\sh}{\mathcal{S(H)}} 
\newcommand{\ip}[2]{\left\langle \,#1\, | \,#2\, \right\rangle} 

\newcommand{\kb}[2]{\left|#1\,\right\rangle \left\langle \,#2 \right|} 
\newcommand{\no}[1]{\left\|#1\right\|} 
\newcommand{\tr}[1]{\textrm{tr}\, (#1)} 

\newcommand{\N}{\mathbb N} 
\newcommand{\Z}{\mathbb Z} 

\begin{document}

\preprint{APS/123-QED}

\title{Relations between convergence rates\\in Schatten $p$-norms}

\author{Paolo Albini}
\email{albini@fisica.unige.it}
\author{Alessandro Toigo}%
\email{toigo@ge.infn.it}
\affiliation{Dipartimento di Fisica, Universit\`a di Genova, Via Dodecaneso 33, 16146 Genova, Italy
}

\author{Veronica Umanit\`a}
\email{veronica.umanita@polimi.it} 
\affiliation{Dipartimento di Matematica ``F. Brioschi'',
Politecnico di Milano, Piazza Leonardo da
Vinci 32, I-20133 Milano, Italy}

\date{\today}

\begin{abstract}
In quantum estimation theory and quantum tomography, the quantum state obtained by sampling converges to the `true' unknown density matrix under topologies that are different from the natural notion of distance in the space of quantum states, i.e.~the trace class norm. In this paper, we address such problem, finding relations between the rates of convergence in the Schatten $p$-norms and in the trace class norm.
\end{abstract}

\maketitle

\section{Introduction}\label{Intro}
In recent years, the advent of quantum tomography has brought a renewed interest for quantum estimation techniques in physics. The underlying physical issue goes back to a well known problem raised by Pauli: for a quantum system, find a set of observables such that the knowledge of the associated probability distributions completely determines the density matrix. Quantum tomography (\cite{VR,Leo,D'Ariano}) solves this problem pointing out a fixed family of observables (quorum) and giving a formula that expresses the state of the system in terms of their probability distributions.
Two questions then arise. First, in many applications the quorum is an infinite set, and it is clear that in practice one can measure only a finite number of observables. Second, and more fundamental, the exact probability distribution of a fixed observable in the quorum can not be measured; experiments yield only an empirical estimate of it, based on a finite number of samples. Thus, the main practical problem is to determine the rate of convergence of the empirical reconstructed state to the `true' unknown state of the quantum system, as the number of samples increases. The natural notion of convergence is with respect to the topology induced by the trace class norm. Indeed, this is the weakest topology that assures strong convergence of the probability distribution of an arbitrary measurement of the empirical state to the probability distribution of the same measurement on the true state, uniformly over all possible measurements. However, since reconstruction formulas involve convergence in norms (like the Hilbert-Schmidt norm) that are weaker than the trace class norm, almost all the results in the literature consider convergence with respect to such weaker norms (\cite{Gill1,Gill2}), usually not giving explicitely the rates. In \cite{Gill2eMezzo,Gill3}, the authors find the rates of convergence, but with respect to a nonmetrisable topology.

A result in \cite{S} shows that on the space of states the trace class and the Schatten $p$-norm topologies with $p>1$ are equivalent. However, in Simon's theorem a quantitative relation between the rates of convergence in the two norms is missing.

In this paper, following the work of \cite{S}, we collect some useful inequalities relating the trace class norm $\no{\cdot}_1$ with the Schatten $p$-norms $\no{\cdot}_p$ for $p>1$. This is intended as a step for passing from the estimate of convergence rates in $p$-norms to the rates in trace class norm.
In our main Theorem \ref{teo}, given trace class operators $A_0$ and $A$, with $A_0$ normal and $\no{A}_1 = \no{A_0}_1 = 1$, we give an estimation of $\no{A-A_0}_1$ with a bound depending only on $\no{A-A_0}_p$ and on the decreasing rate of the eigenvalues $\mu_n$ of $A_0$. In doing so, we have in mind the practical situation in which $A_0$ is the unknown quantum state to be determined, and $A$ is its empirical estimate, over which we impose only a normalisation condition. Manipulating this result, one can find a function $f_{A_0}$, depending only on the $\mu_n$'s, such that $\no{A-A_0}_1\leq f_{A_0}(\no{A-A_0}_p)\no{A-A_0}_p$.
If the decreasing rate of the $\mu_n$'s is fast enough, then the product $f_{A_0}(\no{A-A_0}_p)\no{A-A_0}_p$ converges to $0$ as $A$ approaches $A_0$ with respect to the $p$-norm (see example \ref{es3}). In particular, assuming $A_0$ is chosen in a restricted class of physically feasible states, this allows to find rather explicitly the converging rate in the trace class norm once the $\no{\cdot}_p$-rate is known.

\section{Notations}

Let $\hh$ be a complex Hilbert space, with norm $\no{\cdot}$ and
scalar product $\ip{\cdot}{\cdot}$ linear in the second entry. $\lh$
is the Banach space of bounded operators on $\hh$, with uniform norm
$\no{\cdot}_\infty$. If $A\in\lh$, the modulus of $A$ is $| A | =
(A^\ast A)^{1/2}$.

For $p\geq 1$, we denote by $\bbb{p}$ the Schatten $p$-ideals in $\lh$. We recall that such ideals are the Banach spaces
$$
\bbb{p} = \{ A\in\lh \mid \tr{ |A|^p } < \infty \}
$$
endowed with the norm
$$
\no{A}_p = [\tr{|A|^p}]^{1/p}.
$$
We have $\bbb{p} \subset \bbb{q}$ and $\no{A}_q \leq \no{A}_p$ if $p<q$. In particular, the inclusion $\bbb{p} \hookrightarrow \bbb{q}$ is continuous. The following H\"older inequalities hold
$$
\no{AB}_1 \leq \no{A}_p \no{B}_{p/(p-1)} , \qquad \no{AB}_1 \leq \no{A}_1 \no{B}_\infty .
$$

$\bbb{1}$ is the Banach space of trace class operators on $\hh$. We denote by $\sh$ the set of states in $\bbb{1}$, i.e.~the closed subset of positive trace one elements.

If $A\in\bbb{p}$, then $A$ is a compact operator, so it has a {\em canonical decomposition}
$$
A = \sum_{n \in I} \lambda_n (A) \kb{v_n}{u_n},
$$
where $I = \{ 0,1,2 \ldots \}$ is a finite or countably infinite subset of $\N$,
$\{ v_n \}_{n\in I}$ and $\{ u_n \}_{n\in I}$ are orthonormal sets in $\hh$, and $\lambda_n (A) > 0$. Moreover,
$$
\no{A}^p_p = \sum_{n \in I} \lambda_n (A)^p .
$$
In addition, if $A$ is normal, then the spectral decomposition
$$A = \sum_{n \in I} \mu_n (A) \kb{u_n}{u_n}$$
holds, where $\{\mu_n (A)\}_{n\in I}$ are the nonzero eigenvalues
of $A$ (each eigenvalue appearing in the sequence as many times as
its finite multiplicity) and $\mid\!\mu_n(A)\!\mid=\lambda_n (A)$. 
In particular, if $A$ is positive, then
the spectral and canonical decompositions coincide, i.e. $v_n=u_n$
and $\lambda_n(A)=\mu_n(A)$ are the strictly positive eigenvalues
of $A$.

\section{Main results}
In this section we suppose $1<p<+\infty$. 
\begin{theorem}\label{teo}
Let $A_0,\,A\in\bbb{1}$ with $A_0$ normal and
$\no{A_0}_1 = \no{A}_1 = 1$. Let $\sum_{n\in I} \mu_n ( A_0 )
\kb{u_n}{u_n}$ be the spectral decomposition of $A_0$. Then
\begin{equation}\label{ris}
\no{A_0 - A}_1 \leq 3 N^{(p-1)/p} \no{A_0 - A}_p + 2 \sum_{n\geq
N} \mid\!\mu_n ( A_0 )\!\mid
\end{equation}
for all $N\in\N$.
\end{theorem}

Since some technical lemmas are needed, we postpone the proof of Theorem \ref{teo} to the next section.
\begin{remark} We notice that the hypotheses of the theorem do not ask for either of the two operators to belong to $\sh$. This is convenient, since not all the estimation schemes lead to actual states as estimates: a notable example is the Pattern Function Projection estimator \cite{Gill2}, whose estimates are in $\bbb{1}$ but not necessarily positive.
It is of course very natural to see $A_0 \in \sh$ as the unknown state of the system (since it automatically satisfies all the hypotheses) and the less subjected to hypotheses $A$ as its empirical estimate. In this setting, eq.~(\ref{ris}) expresses the $1$-norm rate of convergence of $A$ to $A_0$ in terms of the $p$-rate and of the decreasing rate of the eigenvalues of $A_0$. Anyway, if we know the estimate $A$ to be normal (a rather common case in the literature) then in eq.~(\ref{ris}) we can exchange the roles of $A_0$ and $A$, thus giving a bound depending on the decreasing rate of the eigenvalues of $A$ instead of $A_0$. We remind the reader that in the general estimation framework we may know nothing about the estimand state and its eigenvalue behaviour, while we actually have hold of the estimator. Thus which of the two possible interpretation is the best (or the only possible) choice depends in general on the case at hand.
\end{remark}

In general one can not exactly compute $\sum_{n\geq
N}\mid\!\mu_n ( A_0 )\!\mid$, but, in several cases, we can
give an estimation of this quantity, as in the following.
\begin{example}\label{es1}
Let $A_0 \in\sh$ with eigenvalues $\lambda_n ( A_0 ) \leq C (n+1)^{-\alpha}$ ($n\in\N$) for some $C > 0$ and $\alpha>1$. Then
$$
\no{A_0 - A}_1 \leq 3 N^{(p-1)/p} \no{A_0 - A}_p + \frac{2}{\alpha-1}\,\frac{1}{N^{\alpha-1}}
$$
for all $A\in\bbb{1}$ with $\no{A}_1 = 1$, $N\geq 1$.\\
Indeed, since the map $x\mapsto (x+1)^{-\alpha}$ is decreasing, we have
$$\sum_{n\geq N}\lambda_n(A_0)\leq C \sum_{n\geq N}\frac{1}{(n+1)^{\alpha}}\leq C \int_{N-1}^{+\infty}\frac{1}{(x+1)^\alpha}dx= \frac{C}{\alpha-1}\,\frac{1}{N^{\alpha-1}},$$ and so the claimed inequality follows by equation (\ref{ris}).
\end{example}
\begin{example}\label{es1,5}
Let us consider $A_0\in\sh$ with eigenvalues
$$
\lambda_n(A_0)\leq Ce^{-\beta n}
$$
and 
$$
\beta>0 , \ C \geq \left(\sum\nolimits_{n\geq 0}e^{-\beta n}\right)^{-1}=1-e^{-\beta}.
$$
Then
\begin{equation}\label{quella}
\no{A_0 - A}_1 \leq 3 N^{(p-1)/p} \no{A_0 - A}_p + 2 C\frac{e^{-\beta N}}{1-e^{-\beta}}
\end{equation}
for all $A\in\bbb{1}$ with $\no{A}_1 = 1$, $N\geq 1$. In fact
\begin{equation}\label{conv0}
\sum_{n\geq N}\lambda_n(A_0)\leq C\sum_{n\geq N}e^{-\beta n}=C\left(\frac{1}{1-e^{-\beta}}-\frac{1-e^{-\beta N}}{1-e^{-\beta}}\right)=C\frac{e^{-\beta N}}{1-e^{-\beta}}.
\end{equation}
\smallskip
As a particular case, eq.~(\ref{quella}) applies to the Gibbs state
$$
A_0 = \sum_{n\geq 0}\frac{e^{-\beta n}}{Z}\kb{n}{n}\quad\quad(\beta>0),
$$
with
$$
C = \frac{1}{Z} = \Big[\sum\nolimits_{n\geq 0}e^{-\beta n} \Big]^{-1} = 1-e^{-\beta}.
$$
\end{example}
\smallskip
From the above Theorem, we obtain the following.
\begin{corollary}\label{cor1}
Let $A_0,\,\{ A_n \}_{n \in \Z_+}$ be elements in $\bbb{1}$ with $A_0$ normal and $\no{A_0}_1=\no{A_n}_1=1$, such that $\no{A_0-A_n}_p\to 0$. Then $\no{A_0-A_n}_1\to 0$.
\end{corollary}
\begin{proof}
Fix $\eps > 0$. Choose $N_\eps$ such that $\sum_{n\geq N_\eps} \mid\!\mu_n (A_0)\!\mid < \eps /4$, and $i_\eps$ such that $\no{A_0 - A_i}_p < N_\eps^{(1-p)/p} \eps /6$ for $i \geq i_\eps$. By eq.~\ref{ris}, for $i \geq i_\eps$ we have
$$\no{A_0 - A_i}_1 < 3 N_\eps^{(p-1)/p} \no{A_0 - A_i}_p + \eps /2 < \eps.$$
This proves our claim.
\end{proof}
In particular, this implies that the topologies induced on $\sh$ by $\bbb{1}$ and $\bbb{p}$ coincide.
\begin{remark}
Corollary \ref{cor1} is also a simple consequence of Theorem 2.19 in \cite{S}, noting that  $p$-convergence implies convergence in the strong operator topology of $\{A_n\}_n$ and $\{A^*_n\}_n$. However, in contrast with Theorem \ref{teo}, Simon's result does not give an explicit relation between the rates of convergence in the $p$- and $1$-norms. 
\end{remark}

Using eq.~\ref{ris}, we can obtain an alternative estimation of the distance $\no{ A - A_0 }_1$.
\begin{corollary}
Let $A_0 \in \bbb{1}$ be normal and $\no{A_0} = 1$. For $\eps > 0$ define
$$
N_{A_0} (\eps) := \min \{ N\in\N \mid \sum_{n\geq N} \mid\! \mu_n (A_0) \!\mid < \eps \}.
$$
Then, if $q=p/(p-1)$,
\begin{equation}\label{NA0}
\no{A_0 - A}_1 \leq ( 3 \sqrt[q]{N_{A_0} (\no{A_0 - A}_p)} + 2 ) \no{A_0 - A}_p ,
\end{equation}
for all $A\in\bbb{1}$ with $\no{A}_1 = 1$.
\end{corollary}

Note that $\lim_{\eps \to 0} N_{A_0} (\eps) < \infty$ if and only if $A_0$ has finite rank.\smallskip\\

The above inequality is a less strict result than the bound in eq.~\ref{ris} and its consequence in Corollary \ref{cor1}. In fact, we show in the next example that one can fix $A_0 \in \sh$ such that,
for any sequence $\{A_n \}_{n\in\Z_+}$ in the unit ball in $\bbb{1}$ with $\no{A_0 - A_n}_p \to 0$, one has $( 3 \sqrt[q]{N_{A_0} (\no{A_0 - A}_p)} + 2 ) \no{A_0 - A}_p \to \infty$ (while, by Corollary \ref{cor1}, $\no{A_0 - A_n}_1 \to 0$).

\begin{example}
Fix $\alpha$ with $1 < \alpha < 2-1/p$. Suppose $A_0 \in \sh$ with $\lambda_n (A_0) =
C (n+1)^{-\alpha}$ ($n\in\N$), where $C = \left[ \sum_{n\in\N} (n+1)^{-\alpha} \right]^{-1}$.
Then,
$$
\sum_{n \geq N} \lambda_n (A_0) = C \sum_{n \geq N}\frac{1}{(n+1)^{\alpha}} \geq C \int^{+\infty}_N \frac{1}{(x+1)^{\alpha}} dx = \frac{C}{\alpha - 1} (N+1)^{1 - \alpha},
$$
so that
$$
N_{A_0} (\eps) \geq \left( \frac{C}{\eps (\alpha - 1)} \right)^{\frac{1}{\alpha - 1}} - 1 \qquad\mbox{for all $0<\eps< C/(\alpha - 1)$.}
$$
Therefore, if $\{A_m\}_{m\in\Z_+}$ is a sequence in the unit ball in $\bbb{1}$ such that $\no{A_0 - A_m}_p \to 0$, we get
$$
\lim_m( 3 \sqrt[q]{N_{A_0}(\no{A_0 - A_m}_p)} + 2 ) \no{A_0 - A_m}_p = +\infty .
$$
\end{example}

However, if some hypotheses are made about the decreasing rate of the eigenvalues of $A_0$, then $( 3 \sqrt[q]{N_{A_0} (\no{A_0 - A_m}_p)} + 2 ) \no{A_0 - A_m}_p \to 0$ for $\no{A_0 - A_m}_p \to 0$, as the following example shows.

\begin{example}\label{es3}
Suppose $A_0 \in \sh$ as in Example \ref{es1,5}.
Since $Ce^{-\beta N}/(1-e^{-\beta})<\eps$ if and only if $N>-\beta^{-1}\ln (C^{-1}\eps(1-e^{-\beta}))$, inequality (\ref{conv0}) implies that
$$
N_{A_0}(\eps)\leq \frac{1}{\beta}\,\ln\!\left(\frac{C}{\eps(1-e^{-\beta})}\right)+1 \qquad\mbox{for all $0<\eps<1$.}
$$
Therefore, if we have a sequence $\{A_m\}_{m\in\Z_+}$ in the unit ball in $\bbb{1}$ with $\no{A_0 - A_m}_p \to 0$ for some $p>1$, we obtain
$$
\lim_m\, (3\sqrt[q]{N_{A_0}(\no{A_0-A_m}_p)}+2)\no{A_0-A_m}_p=0 .
$$
\end{example}

\section{Proof of Theorem \ref{teo}}
\begin{lemma}\label{=} Let $A\in\bbb{1}$ and $P$ be a projection in $\hh$. If $Q=I-P$, then
$$\no{PAP}_1+\no{QAQ}_1=\no{PAP+QAQ}_1.$$
\end{lemma}
\begin{proof}
Let $PAP=\sum_{n\in I}\lambda_n(PAP)\kb{v_n}{u_n}$ and $QAQ=\sum_{n\in J}\lambda_n(QAQ)\kb{z_n}{w_n}$ be the canonical decompositions for $PAP$ and $QAQ$ respectively.\\
Since $\{u_n\}_{n\in I}$, $\{v_n\}_{n\in I}$ are orthonormal sets in $P(\hh)$, and $\{w_n\}_{n\in J}$, $\{z_n\}_{n\in J}$ are orthonormal sets in $Q(\hh)=\ker P$, it follows that $\{u_n,\, w_n\}_{n\in I\cup J}$ is an orthonormal part in $\hh$, and then
$$PAP+QAQ=\sum_{n\in I\cup J}\left(\lambda_n(PAP)\kb{v_n}{u_n}+\lambda_n(QAQ)\kb{z_n}{w_n}\right)$$
is the canonical decomposition for $PAP+QAQ$. This means that $\lambda_n(PAP+QAQ)=\lambda_n(PAP)+\lambda_n(QAQ)$, and so
$$\no{PAP+QAQ}_1=\no{PAP}_1+\no{QAQ}_1.$$
\end{proof}
\begin{lemma}\label{PAP}
Let $A\in\bbb{1}$ and $\{P_i\}_{i=1}^N$ be a family of mutually orthogonal projections in $\hh$. Then
$$\no{\sum_{i=1}^NP_iAP_i}_1\leq\no A_1.$$
\end{lemma}
\begin{proof}
Set $B:=\sum_{i=1}^NP_iAP_i$.\\
Let $A=\sum_n\lambda_n(A)\kb{v_n}{u_n}$ and $B=\sum_m\lambda_m(B)\kb{z_m}{w_m}$ be the canonical decomposition of $A$ and $B$ respectively. Then
$$\lambda_m(B)=\ip{z_m}{Bw_m}=\sum_{i=1}^N\ip{P_iz_m}{AP_iw_m}=\sum_n\lambda_n(A)\alpha_{nm}$$
with $\alpha_{nm}:=\sum_{i=1}^N\ip{P_iz_m}{v_n}\ip{u_n}{P_iw_m}$.
Note that, by Cauchy-Schwartz inequality, we have
\begin{eqnarray*}
\sum_n\mid\!\alpha_{nm}\!\mid&\leq&\left(\sum_{i,n}\mid\! \ip{P_iz_m}{v_n}\!\mid^2\right)^{1/2}\left(\sum_{i,n}\mid\! \ip{u_n}{P_iw_m}\!\mid^2\right)^{1/2}\\
&\leq&\left(\sum_{i=1}^N\no{P_iz_m}^2\right)^{1/2}\left(\sum_{i=1}^N\no{P_iw_m}^2\right)^{1/2}\leq \no{z_m}\,\no{w_m}=1,
\end{eqnarray*}
for $\{v_n\}_n$ and $\{u_n\}_n$ are orthonormal sets and the projections $P_i$ are mutually orthogonal. Similarly, $\sum_m\mid\!\alpha_{nm}\!\mid\leq 1$.\\
Therefore we get
\begin{eqnarray*}
\no{B}_1&=&\sum_m\mid\!\lambda_m(B)\!\mid=\sum_m\mid\!\sum_n\alpha_{nm}\lambda_n(A)\!\mid\leq\sum_m\sum_n\mid\!\alpha_{nm}\!\mid\lambda_n(A)\\
&=&\sum_n\left(\sum_m\mid\!\alpha_{nm}\!\mid\right)\lambda_n(A)\leq\sum_n\lambda_n(A)=\no A_1.
\end{eqnarray*}
\end{proof}
As a simple consequence of the previous Proposition and Lemma \ref{=} we have the following
\begin{lemma}\label{cor}
If $A\in\bbb{1}$, $P$ is a projection in $\hh$ and $Q=I-P$, then
$$\no{PAP}_1+\no{QAQ}_1\leq\no A_1.$$
In particular, if $P$ commutes with $A$, then
$$\no{PAP}_1+\no{QAQ}_1=\no A_1.$$
\end{lemma}
\begin{proof} The second part follows by Lemma \ref{=} since the commutation between
$A$ and $P$ clearly implies $PAP+QAQ=A$.
\end{proof}
\smallskip
We are now in the condition to prove Theorem \ref{teo}.

\begin{proof}[Proof of Theorem \ref{teo}]
Let $P_N = \sum_{n=0}^{N-1} \kb{u_n}{u_n}$ be the projection on the linear span of the first $N$ eigenvectors of $A_0$, and $Q_N = I - P_N$. By triangle inequality
$$
\no{A_0 - A}_1 \leq \no{P_N (A_0 - A)}_1 + \no{Q_N (A_0 - A) P_N}_1 + \no{Q_N A_0 Q_N}_1 + \no{Q_N A Q_N}_1
$$
By Lemma \ref{cor}
\begin{eqnarray}\label{questa}
\no{Q_N A Q_N}_1 &\leq& 1 - \no{P_N A P_N}_1 \\
\no{P_N A_0 P_N}_1 &=& 1 - \no{Q_N A_0 Q_N}_1 , \label{l'altra}
\end{eqnarray}
since $P$ commutes with $A_0$. Therefore, by triangle inequality
\begin{eqnarray*}
\no{P_N A P_N}_1 & \geq & - \no{P_N (A_0 - A) P_N}_1 + \no{P_N A_0 P_N}_1 \\
& = & - \no{P_N (A_0 - A) P_N}_1 + 1 - \no{Q_N A_0 Q_N}_1
\end{eqnarray*}
we have
\begin{eqnarray*}
\no{A_0 - A}_1 & \leq & \no{P_N (A_0 - A)}_1 + \no{Q_N (A_0 - A) P_N}_1 + \no{P_N (A_0 - A) P_N}_1 \\
&& + 2 \no{Q_N A_0 Q_N}_1 \\
& \leq & 3 \no{P_N}_{p/(p-1)} \no{A_0 - A}_p + 2 \no{Q_N A_0
Q_N}_1.
\end{eqnarray*}
Since $\no{P_N}_{p/(p-1)} = N^{(p-1)/p}$ and $\mid\!Q_N A_0
Q_N\!\mid=\sum_{n\geq N}\mid\!\mu_n(A)\!\mid\kb{u_n}{u_n}$ by
spectral theorem, we obtain
$$\no{Q_N A_0 Q_N}_1=\sum_{n\geq
N}\mid\!\mu_n(A)\!\mid$$ so that eq.~\ref{ris} follows.
\end{proof}

\begin{remark}
If in Theorem \ref{teo} it is assumed $A_0 , A\in\sh$, then the above proof simplifies, since Lemma \ref{cor} is no longer needed to prove eqs.~\ref{questa}, \ref{l'altra}. In fact, in this case $\no{A}_1 = \tr{A} = \tr{P_N A P_N} + \tr{Q_N A Q_N} = \no{P_N A P_N}_1 + \no{Q_N A Q_N}_1$.
\end{remark}

\end{document}